\documentclass{sig-alternate}

\usepackage{array}
	
\usepackage{mdwmath}
\usepackage{mdwtab}

\usepackage{graphicx}
\usepackage{float}
\usepackage{fixltx2e}
\usepackage{stfloats}
\usepackage{lipsum}
\usepackage{tabularx}
\usepackage[tight,footnotesize]{subfigure}
\usepackage{caption}
\usepackage{epstopdf}

\usepackage{amsmath}
\usepackage{amssymb}
\usepackage{paralist}
\usepackage{mathrsfs}
\usepackage{url}
\usepackage{mathbbol}
\usepackage{times}
\usepackage{booktabs}
\usepackage{verbatim}

\usepackage{bm}			

\newtheorem{theorem}{Theorem}[section]

\newtheorem{definition}{Definition}
\newtheorem{lemma}[theorem]{Lemma}

\usepackage{color}

\begin{document}
\title{Soft Cache Hits and the Impact of Alternative Content Recommendations on Mobile Edge Caching}
\author{Thrasyvoulos Spyropoulos\\Dept. Mobile Communications\\EURECOM,France\\spyropou@eurecom.fr \and Pavlos Sermpezis\\Institute of Computer Science\\ FORTH, Greece\\sermpezis@ics.forth.gr }

\maketitle
\begin{abstract}
Caching popular content at the edge of future mobile networks has been widely considered in order to alleviate the impact of the data tsunami on both the access and backhaul networks. A number of interesting techniques have been proposed, including femto-caching and "delayed" or opportunistic cache access. Nevertheless, the majority of these approaches suffer from the rather limited storage capacity of the edge caches, compared to the tremendous and rapidly increasing size of the Internet content catalog. We propose to depart from the assumption of hard cache misses, common in most existing works, and consider ``soft'' cache misses, where if the original content is not available, an alternative content that is locally cached can be recommended. Given that Internet content consumption is increasingly entertainment-oriented, we believe that a related content could often lead to complete or at least partial user satisfaction, without the need to retrieve the original content over expensive links. In this paper, we formulate the problem of optimal edge caching with soft cache hits, in the context of delayed access, and analyze the expected gains. We then show using synthetic and real datasets of related video contents that promising caching gains could be achieved in practice.
\end{abstract}

\section{Introduction}
In the context of cellular networks, it is widely believed that aggressive densification, overlaying the standard macro-cell network with a large number of small cells (e.g., pico- or femto-cells), is a promising way of dealing with the ongoing data crunch~\cite{Sapountzis-Infocom2016}. As this densification puts a tremendous pressure on the backhaul network, researchers have suggested storing popular content at the ``edge'', e.g., at small cells~\cite{femto}, user devices~\cite{femtoD2D,Hui-offloading,Whitbeck-offloading,Pavlos-Offload2016}, or vehicles acting as mobile relays~\cite{vigneri2016} in order to avoid congesting the capacity-limited backhaul links, and reduce the access latency to such content. 

Local content caching has been identified as one of the five most disruptive enablers for 5G networks~\cite{Bocc2014}, sparking a tremendous interest of academia and industry alike. While caching had been widely studied in peer-to-peer systems and content distribution networks (CDNs)~\cite{borst2010}, 
the number of storage points required in future dense HetNets are many orders of magnitude more than in traditional CDNs (e.g., 1000s small cells per area covered by one CDN server). Therefore, the storage space per local cache must be significantly smaller to keep costs reasonable. Hence, even though studies assuming a large (CDN-type) cache deep inside the core network~\cite{Erman2011} give promising hit ratios, only a tiny fraction of the constantly and exponentially increasing content catalog could realistically be stored at each edge, leading to low ``local'' cache hit ratios~\cite{Paschos-infocom2016,Paschos-misconceptions}.

Additional ``global'' caching gains could be sought by increasing the ``effective'' cache size visible to each user through: (a) small cell overlaps, where each user is in range of multiple cells and caches (e.g., in the femto-caching case~\cite{femto}), (b) collocated users overhearing the same broadcast channel and benefiting from cached content in other users' caches (as in coded caching~\cite{Ali2014}), and (c) delayed content access, where a user might wait up to a TTL for its request, during which time more than one (fixed~\cite{Pavlos-Offload2016} or mobile~\cite{Hui-offloading,Whitbeck-offloading,vigneri2016}) caches can be seen. These ideas could theoretically increase the cache hit ratio significantly, when the ``global'' cache size becomes large enough  (e.g., when, in the latter example, the aggregate size of all caches a user sees within a TTL becomes comparable to the content catalog). Nevertheless, in most practical cases a local edge cache would realistically fit at most $10^{-3}/10^{-4}$ of the catalog (e.g., just the entire Netflix catalogue is about 3PBs). Even if the above methods offered a $10\times$ effective cache increase, they would not suffice to achieve significant cache hit ratios (e.g., in the notation of~\cite{Ali2014}, the key factor $KM/N$ would be equal to $10^{-2}$, leading to a global caching gain of $\frac{1}{1+10^{-2}}$, a mere $1\%$ of extra gain).   

Operators, are thus left with a very costly dilemma: bear a huge cost for the backhaul infrastructure (e.g., fiber everywhere) or bear a huge cost for CDN-like storage at each and every small cell. We believe this dilemma stems from the common underlying assumption of almost every caching scheme to try to satisfy \emph{every} possible user request, either from the local cache or, in the worst case, the content server. This leads to an immense catalogue of potential content. Our main assertion in this paper is that, in an Internet which is becoming increasingly content-centric and entertainment-oriented, a radically different approach could be beneficial, namely \emph{moving away from satisfying a given user request towards satisfying the user.} E.g., a user requesting a content X, not available locally (e.g., a fan wanting to follow last weekend's premier league's games), might be equally satisfied (in the best case) or not fully dissatisfied (in many cases), if she receives another content Y related to X (e.g., another premier league game from that weekend). Another example is users streaming content \emph{in sequence} (e.g., browsing YouTube videos back-to-back or listening to personalized radio). In that case, the selected content at each step is often \emph{recommended related to the previous one}, and the user might be almost equally happy with many alternatives. 
We will use the term \emph{soft cache hit} to describe such scenarios.
Finally, we believe such a system is timely given the recent interest of content providers with sophisticated recommendation engines, such as NetFlix and YouTube (i.e., Google), to act as Mobile Virtual Network Operators (MVNO) in the context of RAN Sharing~\cite{CellSlice2013}. 

To this end, we perform here a preliminary analysis and performance evaluation of such a system, in order to obtain initial insights. 
We first formulate the problem of edge caching with \emph{soft cache hits}, and analyze the expected gains. We then show using both synthetic data and a real dataset of related video contents that interesting caching gains could be achieved in practice. Our problem formulation and analysis takes place in the context of \emph{delayed content access} via static or mobile small cells~\cite{vigneri2016,Pavlos-Offload2016}, for two reasons: (a) we believe such delayed access is interesting for low-cost users (e.g., 2 euro plans for operators like Free~\cite{free}) or developing regions, and (b) could be easily combined with soft cache hits to achieve multiplicative gains. Nevertheless, the basic tenets of our approach are equally applicable to femto-caching (i.e., the framework of~\cite{femto}) or even other PHY-aware caching systems~\cite{Psounis2015}. 

To the best of our knowledge, the closest related work to the idea of soft cache hits is Roadcast~\cite{Roadcast}, proposing a query-response based P2P VANET system, where users' query requirements can be relaxed in order to get a matching response sooner. Nevertheless, this work focuses mostly on content similarity metrics and considers heuristics to achieve a square root based allocation policy, known to be optimal in P2P systems. Square root policies are suboptimal in our problem setup, as proven later, with or without soft cache hits~\cite{vigneri2016}.

\section{Problem Setup}\label{sec:model}

\emph{Content Model}: We consider a wireless network with randomly distributed users, requesting contents from a catalogue $\mathcal{K}$ with $\| \mathcal{K} \| = K$ contents. A user requests content $i \in \mathcal{K}$ with probability $p_i$. Without loss of generality (``w.l.o.g.'') we assume all contents have the same size.

\emph{Network Model}: Our network consists of $M$ small cells (SC). These SCs can be either static (as in the femto-caching model~\cite{femto}) or mobile (e.g. a vehicular cloud as in~\cite{vigneri2016}). We denote the set of all SCs as $\mathcal{M}$. We also assume that each SC is equipped with storage capacity of $C$ contents. Accessing content directly from the local cache, i.e. a \emph{cache hit}, is considered ``cheap'' while a \emph{cache miss} leads to an ``expensive'' access (e.g. of the backhaul link in~\cite{femto} or the macro-cell in~\cite{vigneri2016}). 

\emph{Delayed Access Protocol}: If the requested content is not available in a nearby small cell, the user waits until it encounters other small cells (as a result of user or cell mobility), until a Time-To-Live $T$. If the content is not found in any SC within $T$, a cache miss occurs and the content is fetched over the expensive link. 

\emph{Meeting Model}: Meetings between each user and each SC are IID, with the \emph{residual} time until such a meeting occurs being a random variable with CDF $F(t)$. 
\begin{lemma}\label{lemma:pmiss}
If there are $N$ total SCs storing the requested content, the probability of not encountering any of them within $T$ is 
\begin{equation}\label{eq:pmiss}
P_{miss}(N) = \overline{F(T)}^N
\end{equation}
\end{lemma}
The above result follows directly from the definition of $F(t)$ and the assumption of IID meetings.

For simplicity, in this paper we will focus on $F(t) = 1-\exp^{-\lambda t}$, so that $P_{miss}(N) = \exp^{- \lambda N t}$. The identical meeting rates assumption can be further relaxed, as explained in Section~\ref{sec:discussion}.


Up to this point, the problem setup is the same as in~\cite{vigneri2016,Pavlos-Offload2016}. The main departure from that model is captured in the following.

\emph{Content Relation Graph}: Each content $i \in \mathcal{K}$ has a set of \emph{related contents}. Let $u_{ij}$ denote the utility a given user gets if she originally asks for content $i$ but instead receives content $j$, where $0 \le u_{ij} \le 1$ and $u_{ii} = 1, \forall i$. The set of related contents $\mathcal{R}_{i} \subseteq \mathcal{K}$ can be formally defined as: $\mathcal{R}_{i} =\{j \in \mathcal{K}: j \ne i, u_{ij} > 0\}$. These relations define a content relation matrix (or graph) $\mathbf{U} = \{u_{ij}\}$. 

\emph{Delayed Access with Soft Cache Hits (SCH)}: A user again performs delayed access. However, if the requested content $i$ is not found within $T$, but a content in $j \in \mathcal{R}_{i}$ is found in one of the encountered caches, a soft cache hit occurs (and thus no expensive access is needed). A cache miss occurs if neither the requested nor any related content is found within $T$, in which case the original content is retrieved over the expensive link. The soft cache hit utility is equal to $u_{ij}$. We will consider two main cases for $\mathbf{U}$. 
\begin{itemize}
\item \emph{Soft Cache Hits (Case 1):} $u_{ij} = 1, \forall j \in \mathcal{R}_{i}$. Any related content gives a cache hit. As soon as one is found, the user stops looking.
\item \emph{Soft Cache Hits (Case 2):} $u_{ij} = c \; (0 < c < 1), \forall j \in \mathcal{R}_{i}$. If a related content $j \in \mathcal{R}_{i}$ is found before $T$, the user now continues looking for $i$ until $T$. If it fails, a \emph{soft cache hit} occurs and the access to the expensive link is still avoided. However, the utility attained is less than 1 (equal to $c$), which creates an interesting tradeoff. If neither $i$ nor any related $j$ is found by $T$, then a cache miss occurs, as usual. 
\end{itemize}


\section{Caching with Related Content}\label{sec:theory}
\subsection{Objectives}

The goal in the above defined problem is to minimize the number of bytes accessed over the expensive ``link'' (which is, as explained, a radio access link to a macro-cell and/or the backhaul network). When all contents have the same size, this is simplified to minimizing the number of (expensive) accesses, or equivalently, \emph{maximizing the cache hit ratio}. 

\begin{definition}[Feasible Placement]\label{def:feasible}
Let $N_{i}$ denote the number of SC caches storing content $i$. A placement vector $\mathbf{N} = \{N_{1},\dots,N_{K}\}$ is ``feasible'', if it satisfies the following constraints:
\begin{eqnarray}
0 \le & N_i & \le M,  \label{eq:const-cells} \\ 
\sum_{i=1}^{K} N_{i} & \le & M \cdot C. \label{eq:const-total}
\end{eqnarray}
\end{definition}
$N_{i}$ are the main optimization variables for our problem. Constraint (\ref{eq:const-cells}) says that the number of SCs storing content $i$ is non-negative and at most equal to the total number of SCs, and constraint (\ref{eq:const-total}) that the total number of content replicas stored at all the edge caches cannot exceed their total capacity.

In the traditional case of delayed access no soft cache hits are allowed. This will serve as our \emph{baseline} scenario. The problem objective (i.e., the expected hit ratio) in this case is given in the following lemma.

\begin{lemma}[Cache Hit Ratio - Base]\label{lemma:base-objective}
Assume a feasible placement vector $\mathbf{N}$. The cache hit rate, i.e., the expected number of user requests served locally when no soft cache hits are allowed is equal to
\begin{equation} \label{eq:obj-base}
g_{Base}(\mathbf{N})  = \sum_{i=1}^{K} p_{i} \cdot \left(1-e^{-\lambda \cdot T \cdot N_{i}}\right). 
\end{equation}
\end{lemma}
The objective (Eq.(\ref{eq:obj-base})) in the above lemma is straightforward in light of Lemma~\ref{lemma:pmiss} and the model of Section~\ref{sec:model}. 

As explained earlier, when we do allow soft cache hits, if content $i$ is requested, a cache hit can occur also if other contents $j$ (related to $i$) can be accessed on time. The modified objective for Cases 1 and 2 of the content relation graph $\mathbf{U}$ is given in the following two lemmas (the proofs are based on basic probabilistic arguments, and are omitted for brevity).
 
\begin{lemma}[Soft Cache Hit Ratio (Case 1)]\label{lemma:obj-case1}
Assume a feasible placement vector $\mathbf{N}$, and a content relation graph $\mathbf{U}$, where $u_{ij} \in \{0,1\}, \forall i,j \in \mathcal{K}$. The cache hit rate for $\mathbf{N}$ is equal to 
\begin{equation}\label{eq:obj-case1}
g_{SCH1}(\mathbf{N}) = \sum_{i=1}^{K} p_{i} \cdot \left(1- e^{-\lambda \cdot T \cdot \sum_{j = 1}^{K} N_{j} \cdot u_{ij} }\right)
\end{equation}
\end{lemma}

\begin{lemma}[Soft Cache Hit Ratio (Case 2)]\label{lemma:obj-case2}
Assume a feasible placement vector $\mathbf{N}$, and a content relation graph $\mathbf{U}$, where $u_{ii} = 1, \forall i$, and $u_{ij} \in \{0,c\}, \forall j \in \mathcal{K}\backslash \{ i \} $. The cache hit rate for $\mathbf{N}$ is equal to 
\begin{align}
g_{SCH2}(\mathbf{N}) & = \sum_{i=1}^{K} p_{i} \cdot \Big[ \left(1-e^{-\lambda \cdot T \cdot N_{i}}\right) \nonumber \\ 
& + c \cdot e^{-\lambda \cdot T \cdot N_{i}} \cdot \left(1-e^{-\lambda \cdot T \cdot \sum_{j \in R_{i}} N_{j}} \right)\Big] \label{eq:obj-case2}
\end{align}
\end{lemma}

The main difference between these two cases is that, in the first case, finding a related content gives utility $1$ and is equivalent to a normal cache hit. However, in the second case, a related content allows the operator to avoid accessing the expensive link, but is penalized because the utility for the user is lower, leading to a utility of $c < 1$ (we remind the reader that $\mathcal{R}_{i}$ in the second term of Eq.(\ref{eq:obj-case2}) includes all related contents $j$, such that $u_{ij} > 0$, but does not include content $i$).

\subsection{Performance Improvement Under the\\ Baseline Placement}

Maximizing the objective of Lemma~\ref{lemma:base-objective} within the feasibility region of Definition~\ref{def:feasible}, defines the optimal cache allocation problem for the baseline scenario (no soft cache hits). This is in general an INLP (Integer Non-Linear Program) that relates to a ``multiple knapsack'' problem (with equal capacities and logarithmic rather than linear utilities) and is NP-hard to solve. Various polynomial approximation algorithms exist with good performance when the size of the caches are large enough to fit many contents. One such approximation can be achieved by solving a continuous relaxation of the problem (related to the fractional knapsack problem), where the optimization variables $N_{i} \in [0,M]$ are continuous. In that case, it is easy to show that the baseline problem is convex, whose optimal solution can be found analytically using Lagrangian multipliers and solving the KKT conditions (we refer the interested reader to~\cite{vigneri2016} for more details). Specifically, the optimal solution is given by
\begin{equation}\label{eq:optN-base}
N_i^{*} = \left\{ 
	\begin{array}{l l}
		0,	& \quad \text{if } p_i < L \\
		\frac{1}{\lambda T} \ln \left(\frac{p_{i} \lambda T}{\rho}\right), & \quad \text{if } L \leq p_i \leq U \\
		M,	&\quad \text{if } p_i > U
	\end{array} \right.
\end{equation}
\noindent where $L\triangleq \rho \cdot(\lambda T)^{-1}$, $U\triangleq \rho \cdot(\lambda T)^{-1} \cdot e^{\lambda \cdot M \cdot T}$, and $\rho$ is an appropriate Lagrange multiplier corresponding to the capacity constraint of Eq.(\ref{eq:const-total}).\footnote{An integer solution could be obtained by rounding~\cite{vigneri2016,femto}. Alternatively, one could interpret a non-integer $N_{i}$ value as follows: If $N_{i}=7.6$, $100\%$ of content $i$ is allocated to $7$ caches, and one more cache stores only $60\%$ of the content. If a user encounters the latter, she retrieves the remaining $40\%$ from the infrastructure.} 

Replacing $N_{i}^{*}$ in the objective of the baseline problem (Eq.(\ref{eq:obj-base})) gives us the optimal cache hit ratio, if we ignored related content. At the same time, replacing $N_{i}^{*}$ in the objective of Eq.(\ref{eq:obj-case1}) gives us the cache hit ratio when we can satisfy a request with related content, \emph{but the caching decisions were already taken and are the original ones}. 
(We will show later that we could do even better by considering the related content graph $\mathbf{U}$ when solving the cache placement problem.) The following theorem provides the expected improvement in terms of load on the expensive link, for a simple scenario where $L \le p_{i} \le U, \forall i \in \mathcal{K}$.

\begin{theorem}
Assume that $\|\mathcal{R}_{i}\| = L, \forall i \in \mathcal{K}$. The expected improvement in the cache hit ratio by recommending alternative contents, when the optimal cache placement algorithm is oblivious to these recommendations, is equal to
\begin{equation}\label{eq:gain-case1}
\frac{1-g_{Base}(\mathbf{N^{*}})}{1-g_{SCH1}(\mathbf{N^{*}})} = K \cdot \left(\frac{\lambda T}{\rho}\right)^{L-1} \frac{1}{\sum_{i \in \mathcal{K}} p_{i} \cdot \Pi_{j \in  \mathcal{K}} \frac{1}{p_{j}^{u_{ij}}}}
 \end{equation} 
 \end{theorem} 
 \begin{proof}
 The cache miss ratio (or ``load'' on the main infrastructure) in the baseline problem is $1-g_{Base}(\mathbf{N^{*}})$. Replacing Eq.(\ref{eq:optN-base}) into Eq.(\ref{eq:obj-base}) gives
\begin{align}
1-g_{Base}(\mathbf{N^{*}}) & \overset{Eq.(\ref{eq:obj-base})}{=} \sum_{i=1}^{K}  p_{i} \cdot e^{-\lambda \cdot T \sum_{j=1}^{K} \cdot N_{i}^{*}} \nonumber \\ 
\overset{Eq.(\ref{eq:optN-base})}{=}   \sum_{i=1}^{K}  &   p_{i} \cdot e^{ \ln \left( \frac{\rho}{p_{i} \lambda T} \right)} = \sum_{i=1}^{K}  p_{i} \cdot \frac{\rho}{p_{i} \lambda T} = \frac{K \rho}{\lambda T} 
\label{eq:base-miss-rate}
\end{align} 
Similarly, let's assume that an original request could be satisfied with a related content as in Lemma~\ref{lemma:obj-case1}. The cache miss ratio, denoted as $1-g_{SCH1}(\mathbf{N^{*}})$, can be calculated as:
 \begin{align*}
&1-g_{SCH1}(\mathbf{N^{*}}) \overset{Eq.(\ref{eq:obj-case1})}{=}  \sum_{i=1}^{K}  p_{i} \cdot e^{-\lambda \cdot T \cdot \sum_{j = 1}^{K} N^{*}_{j} \cdot u_{ij} } \\
& \overset{Eq.(\ref{eq:optN-base})}{=}  \sum_{i=1}^{K}  p_{i} \cdot e^{-\lambda \cdot T \cdot \left( \sum_{j = 1}^{K} \frac{1}{\lambda T} \ln \left(\frac{p_{j} \lambda T}{\rho}\right) \cdot u_{ij} \right) } \\
&  = \sum_{i=1}^{K}  p_{i} \cdot e^{ \sum_{j = 1}^{K} \ln \left(\frac{\rho}{p_{j} \lambda T}\right) \cdot u_{ij}} \\
&  = \sum_{i=1}^{K}  p_{i} \cdot \Pi_{j \in  \mathcal{K}} \left( \frac{\rho}{\lambda T} \frac{1}{p_{j}} \cdot u_{ij} \right) = \left( \frac{\rho}{\lambda T} \right)^{L} \sum_{i=1}^{K}  p_{i} \cdot \Pi_{j \in  \mathcal{K}}  \frac{1}{p_{j}^{u_{ij}}}
\end{align*} 
Hence, the gain from soft cache hits (case 1) is equal to $\frac{1-g_{Base}(\mathbf{N^{*}})}{1-g_{SCH1}(\mathbf{N^{*}})}$, which gives the desired Eq.(\ref{eq:obj-base}).
\end{proof}
 
The case where some contents receive no or maximum ($M$) copies, as in Eq.(\ref{eq:optN-base}), can be easily derived by modifying the summation in the above proofs. As a very simple example, consider the case of uniform content popularity, i.e. $p_{i} = \frac{1}{K}$. After some simple calculations, we get that the performance benefits by related content are equal to $\left( \frac{K \rho}{\lambda T}  \right)^{-(L-1)}$. However, we know that $\frac{K \rho}{\lambda T} \le 1$, since it is the cache miss rate of the base policy (see Eq.(\ref{eq:base-miss-rate})). Therefore, the above gain  $\left( \frac{K \rho}{\lambda T}  \right)^{-(L-1)} \ge 1$, and is increasing in $L-1$, the number of related contents per content $i$, as one would expect.
A similar result can be easily derived for Case 2, as well as when the number of non-zero elements on each row $i$ of $\mathbf{U}$ is different (i.e. not all equal to $L$). 
 
\subsection{Content Graph Aware Optimal Caching}
 
We have so far assumed that the caching policy is unaffected by the ability to recommend alternative contents. While this already leads to performance gains, as shown earlier, it is still suboptimal. For example, assume a user requesting content $A$ would be OK to receive instead content $B$ (i.e. $u_{AB} = 1$) and a user requesting content $B$ would be OK to receive content $A$ instead  (i.e. $u_{BA} = 1$). If both contents $A$ and $B$ are popular, \emph{a standard caching policy would give a high number of replicas to both}, according to Eq.(\ref{eq:optN-base}). However, this is clearly suboptimal here, since the caching algorithm could just store only one of the two at each cache, saving valuable capacity that could be used to store other contents. The following two theorems formalize this for the two content relation graph cases, discussed in Section~\ref{sec:model}. Due to space limitations, we only show the proof for the more generic Case 2. 
 
 \begin{theorem}[\textbf{U}-aware Optimal Caching (Case 1)]\label{lemma:opt-case1}
Assume a content relation graph $\mathbf{U}$, where $u_{ij} \in \{0,1\}, \forall i,j \in \mathcal{K}$. The optimal content placement that directly exploits related contents is given by vector $\mathbf{N^{*}_{SCH1}}$ which is the solution to the following optimization problem 
\begin{eqnarray*}
\underset{\mathbf{N}}{\mbox{maximize}} = \sum_{i=1}^{K} p_{i} \cdot \left(1- e^{-\lambda \cdot T \cdot \sum_{j = 1}^{K} N_{j} \cdot u_{ij} }\right), \\ 
\mbox{subject to } \mathbf{N} \mbox{ feasible (according to Definition~\ref{def:feasible})}
\end{eqnarray*}
Furthermore, the above problem is a convex optimization problem.
\end{theorem}



\begin{theorem}[\textbf{U}-aware Optimal Caching (Case 2)]\label{lemma:opt-case2}
The optimal content placement defined by maximizing the objective of Lemma~\ref{lemma:obj-case2}, subject to the feasibility constraints of Definition~\ref{def:feasible}, gives the optimal content allocation vector $\mathbf{N^{*}_{SCH2}}$. Furthermore, the problem is also convex.
\end{theorem}
\begin{proof}
It is easy to see that the feasibility region (Definition~\ref{def:feasible}) is convex. The objective function needs to be concave (since this is formulated as a maximization problem). A sufficient condition is if its Hessian matrix $\textbf{H}$ is negative semi-definite, i.e., $\textbf{z}^{T}\cdot \textbf{H}\cdot \textbf{z}\leq 0,~\forall \textbf{z}=\{z_{i}\}\geq 0$.

Taking the derivatives of the objective function $g_{SCH2}$, we calculate the terms of the Hessian matrix
\begin{align*}
H_{m,m}	=-(\lambda\cdot T)^{2}& \cdot \Big[p_{m}\cdot (1-c)\cdot e^{-\lambda\cdot T\cdot N_{m}} \\&+ \sum_{i=1}^{K} p_{i}\cdot c\cdot \mathcal{I}_{im}\cdot \mathcal{I}_{in}\cdot e^{-\lambda\cdot T\cdot \sum_{j=1}^{K}N_{j}\cdot \mathcal{I}_{ij}}\Big]
\end{align*}
and for $m\neq n$
\begin{align*}
H_{m,n}	
&=-(\lambda\cdot T)^{2} \sum_{i=1}^{K} p_{i}\cdot c\cdot \mathcal{I}_{im}\cdot \mathcal{I}_{in}\cdot e^{-\lambda\cdot T\cdot \sum_{j=1}^{K}N_{j}\cdot \mathcal{I}_{ij}}
\end{align*}
where $\mathcal{I}_{nm}$ is $1$ if $u_{nm}>0$; otherwise is $0$.

Then, the product $\textbf{z}^{T}\cdot \textbf{H}\cdot \textbf{z}$ is given by the expression
\begin{align*}
\textbf{z}^{T}\cdot \textbf{H}\cdot \textbf{z} &= -(\lambda\cdot T)^{2} \sum_{m=1}^{K} \Big[ z_{m}^{2}\cdot p_{m}\cdot (1-c)\cdot e^{-\lambda\cdot T\cdot N_{m}} \\
&+\sum_{n=1}^{K}\sum_{i=1}^{K}z_{m}\cdot z_{n}\cdot p_{i}\cdot \mathcal{I}_{im}\cdot \mathcal{I}_{in}\cdot e^{-\lambda\cdot T\cdot \sum_{j=1}^{K}N_{j}\cdot \mathcal{I}_{ij}} \Big]
\end{align*}
which is always $\leq 0$.
\end{proof}

~\\

\section{Performance Evaluation}\label{sec:sims}
\subsection{Simulations Setup}\label{sec:sims-setup}

\textbf{Mobility Trace.} We use the TVCM mobility model to generate a trace, where nodes move in a square area $1000m\times1000m$ comprising three sub-areas of interest (communities). Each node moves inside its community for 60\% of the time, and leaves it for a few short periods. The area is entirely covered by macro-cell BSs, and also includes 25 non overlapping small-cell base stations (SCs), with a communication range of 100m. 

\textbf{Content Popularity.} We create $K=1000$ contents and assign to each of them a popularity value $p_{i}$ drawn from a Zipf distribution, $p_{i}\in[1,1000]$ with shape parameter $\alpha=2$. Power-law distributions have been shown to capture well real popularity patterns~\cite{youtube-traffic-from-edge,top-video-cellular,pavlos-dataset-AOC}.

\textbf{Utility Matrix.} To investigate the effect of the matrix U, we generate different matrices belonging to two generic classes:

\noindent\textit{(a) random U}: for each content pair $\{i,j\}$, the utility is $u_{ij} = 1$ with probability $p=\frac{L}{K}$ (otherwise it is 0), such that each content has on average $L$ related contents, i.e., $L = E[||\mathcal{R}_{i}||]$.

\noindent\textit{(b) popularity proportional U}: for each content pair $\{i,j\}$, the utility is $u_{ij} = 1$ with probability $p=L'\cdot \frac{\cdot p_{j}}{\sum_{j}p_{j}}$ (otherwise it is 0), where $p_{j}$ is the popularity of content $j$, and $L'$ is a normalization parameter that determines $E[||\mathcal{R}_{i}||]$. 

\textbf{YouTube datasets.} In addition to the synthetic popularity/utility patterns, we use real datasets from YouTube that contain information about \textit{video popularity} and \textit{related video lists}~\cite{youtube-related-videos-dataset}. Table~\ref{table:youtube-instances} contains information about the datasets we use, and some main statistics. We pre-process the data to remove entries with $0$ or no popularity value. For each video $j$ appearing in the related videos list of a video $i$, we set $u_{ij}=1$ and $u_{ji}=1$. Due to the sparseness of the datasets, we consider only the videos belonging to the largest connected component of the graph with vertices $\mathcal{V}=\{i:i\in\mathcal{K}\}$ and edges $\mathcal{E}=\{\epsilon_{ij}:i,j\in\mathcal{K}, u_{ij}=1\}$.


\begin{table}[b]
\caption{YouTube dataset instances information (after processing).}\label{table:youtube-instances}
\begin{tabular}{l|ccc}
{}					& Data (date / depth of search~\cite{youtube-related-videos-dataset})				& $K$			& $E[||\mathcal{R}_{i}||]$\\
\hline
Instance 1 			& 27 July 2008~~~~~~/~~3					&2098			&5.3\\
Instance 2			& 27 March 2008~~/~~1				&1086			&7.9
\end{tabular}
\end{table}

\subsection{Effects of Utility Matrix}
We first study the effects of the (a) density, $L=E[||\mathcal{R}_{i}||]$, and (b) type (\textit{random U} / \textit{popularity proportional U}) of the utility matrix. Specifically, in Fig.~\ref{fig:TVCM-synthU-sims-utility-vs-nb_related} we present the soft-cache hit ratio for the \textit{SCH1} and \textit{SCH2} (with $c=0.5$) cases, under the base optimal policy  $\mathbf{N^{*}}$, as well as the hit ratio of the scenarios without soft caches (\textit{no-soft caches}). As expected, the cache hit rate improves as the density $L$ of the matrix U (x-axis) increases. Under random U matrices the increase in the cache hit rate is almost linear (Fig~\ref{fig:TVCM-synthU-sims-utility-vs-nb_related-random}) on $L$, but quickly plateaus for the popularity proportional U matrices (Fig.~\ref{fig:TVCM-synthU-sims-utility-vs-nb_related-popularity-proportional}). This is reasonable as the achieved cache hit ratios for the popularity proportional U case already reach values $> 90\%$, for few related contents. The reason is that popular contents that have higher probability to appear in the related list of other contents, are also stored in more caches (under the base optimal policy).   

These initial observations show that the performance can be improved by recommending more contents (density) and/or by selecting carefully which contents to recommend (type of matrix U). This is a positive message, since there are more than one degrees of freedom for a system design, allowing thus improvements under various settings (e.g., restriction on the max number of recommended contents, predefined content relations), and enabling cross-layer (application/network) design and optimization approaches.

\begin{figure}
\subfigure[random U]{\includegraphics[width = 0.49\linewidth]{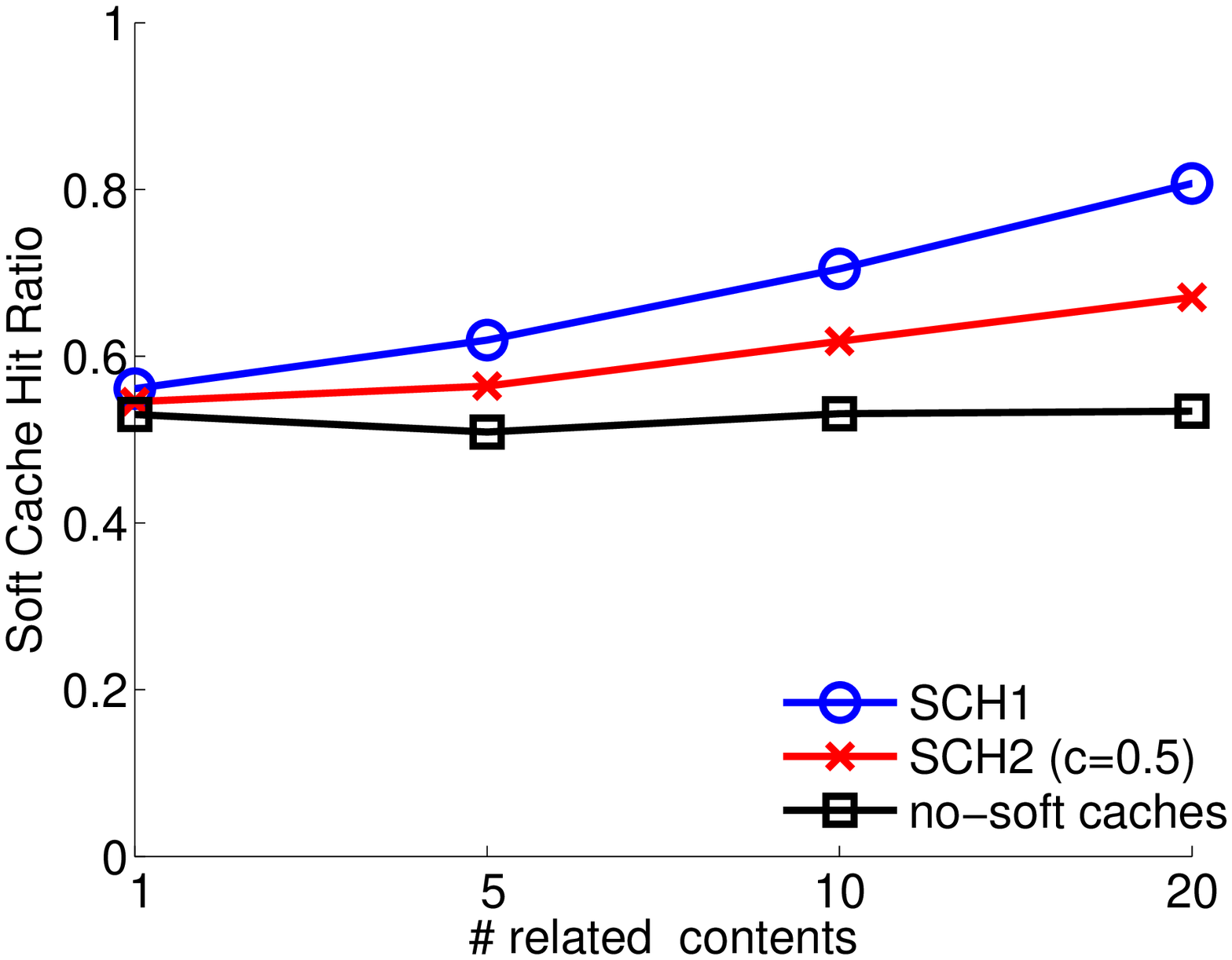}
\label{fig:TVCM-synthU-sims-utility-vs-nb_related-random}}
\subfigure[popularity proportional U]{\includegraphics[width = 0.49\linewidth]{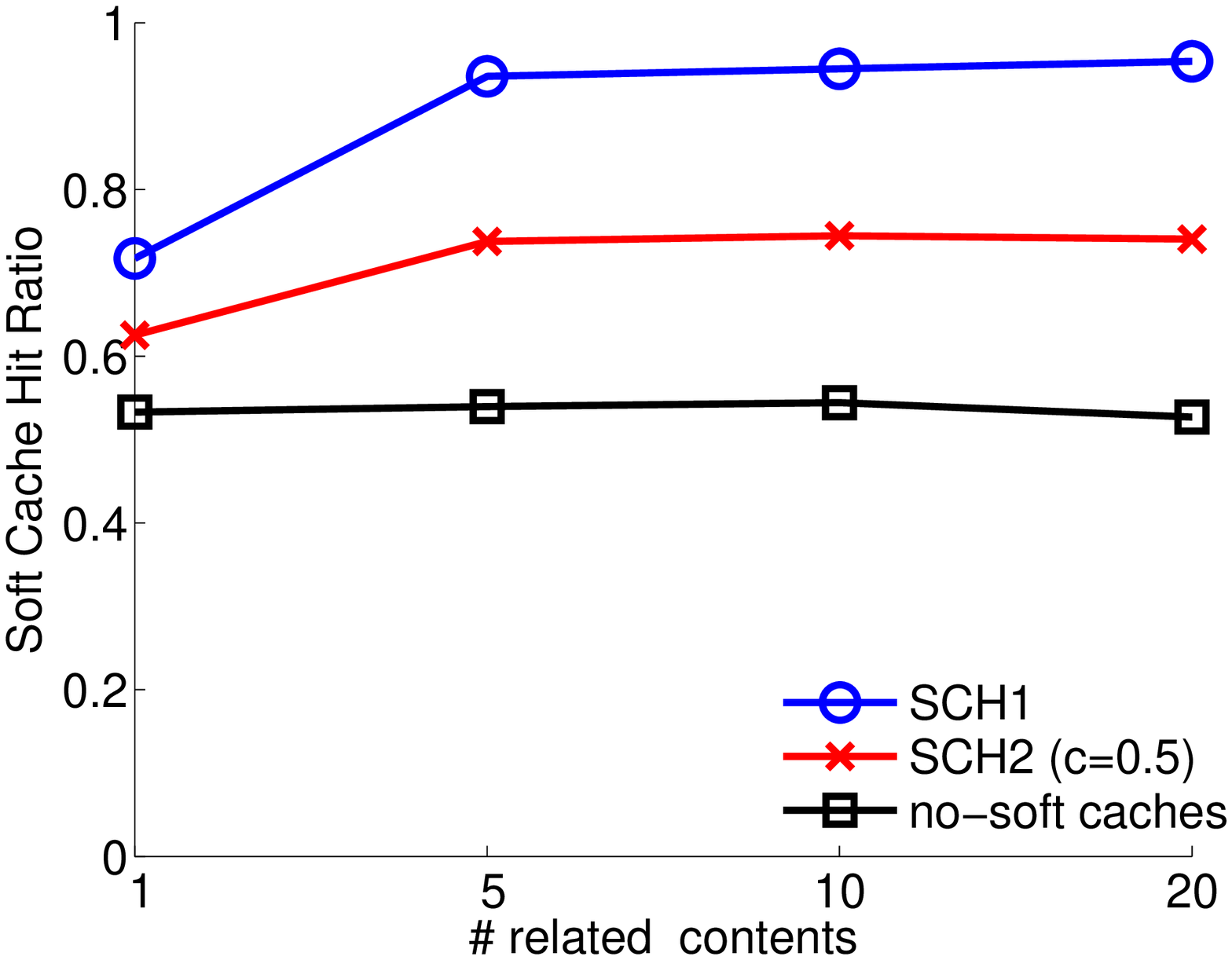}
\label{fig:TVCM-synthU-sims-utility-vs-nb_related-popularity-proportional}}
\caption{Scenarios where the matrices U are generated in (a) random  and (b) popularity proportional way, so that the expected number of related contents per content equals the value of the x-axis. $Q=20$ and $TTL = 5min$}\label{fig:TVCM-synthU-sims-utility-vs-nb_related}
\end{figure}

\subsection{Gains of Optimal Caching Policies}
In Fig.~\ref{fig:TVCM-synthU-sims-bars} we compare the performance gains of the base optimal policy $\mathbf{N^{*}}$ and the U-aware optimal policy  $\mathbf{N^{*}_{SCH1}}$ the \textit{SCH1} case (Theorem~\ref{lemma:opt-case1}). Under random U matrices (Fig.~\ref{fig:TVCM-synthU-sims-bars-random}), the achieved cache hit rate by $\mathbf{N^{*}_{SCH1}}$ is always higher than in the $\mathbf{N^{*}}$ policy, with an increase of $44\%$ (for $TTL=1$min) and $34\%$ (for $TTL=20$min). Here, we need to stress that the extra performance gain from the U-aware optimal caching policy $\mathbf{N^{*}_{SCH1}}$ comes without any cost for the system: the recommendation system (matrix U) and the caching capacity ($M$ and $Q$) remain the same, and only the caching policy changes (i.e., in practice, this corresponds to a simple modification in the content placement algorithm).

In the popularity proportional U case (Fig.~\ref{fig:TVCM-synthU-sims-bars-popularity-proportional}), the performance improvement of the U-aware optimal policy $\mathbf{N^{*}_{SCH1}}$ over the base optimal policy $\mathbf{N^{*}}$ is moderate ($9\%$ and $16\%$, for $TTL=1$min and $TTL=20$min, respectively). This indicates that when a recommendation system is carefully designed for a mobile environment (i.e., in our example, resulting to a popularity proportional matrix U), the U-aware caching policy does not add significant gains. As a result, only the content popularities is needed for the caching placement algorithm. Hence, the \textit{network provider} does not need to cooperate further with a \textit{content provider} (which designs also the recommendation system), e.g., YouTube or Netflix, and this facilitates the deployment of a soft-cache system in practice.

\begin{figure}
\subfigure[random U]{\includegraphics[width = 0.49\linewidth]{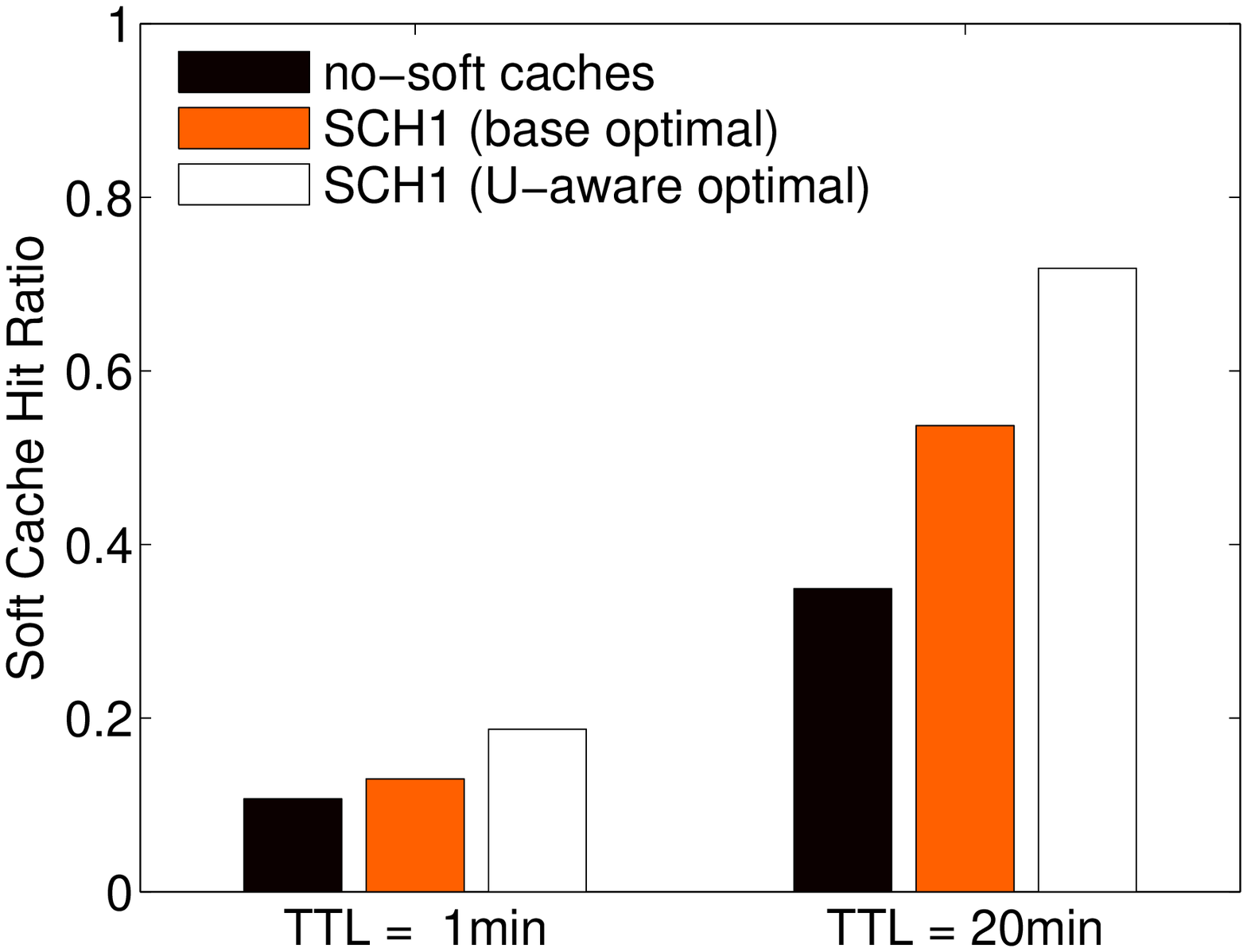}\label{fig:TVCM-synthU-sims-bars-random}}
\subfigure[popularity proportional U]{\includegraphics[width = 0.49\linewidth]{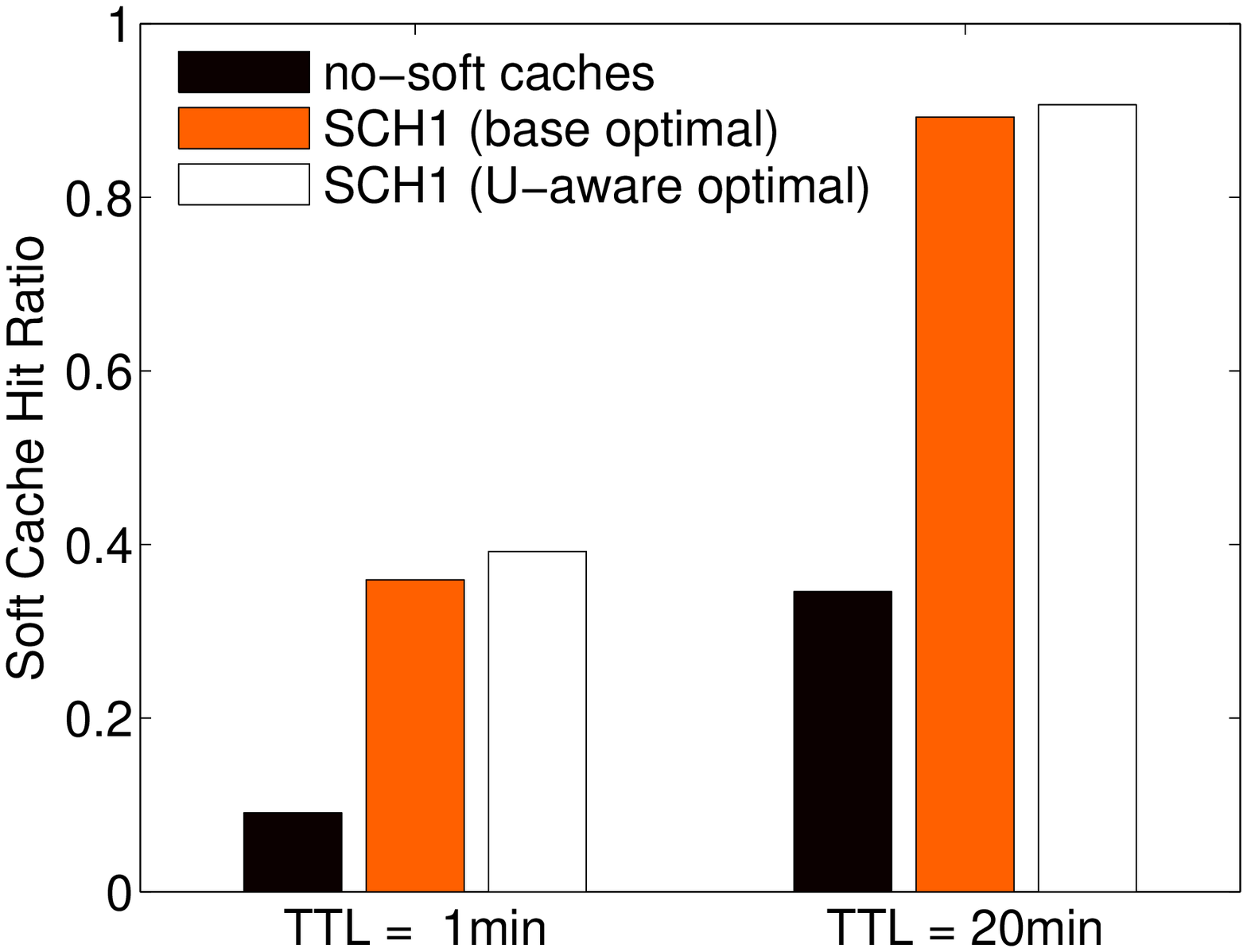}\label{fig:TVCM-synthU-sims-bars-popularity-proportional}}
\caption{Scenarios where the matrices U are generated in (a) random  and (b) popularity proportional way, so that the expected number of related contents per content is $L=5$. The capacity of caches is $Q=5$ (i.e., $0.5\%$ of the catalogue size $M$).}\label{fig:TVCM-synthU-sims-bars}
\end{figure}

\subsection{Gains of the YouTube's Recommendation System}
We conduct simulations on the TVCM mobility trace using the popularity/utility patterns of the YouTube datasets (see Section~\ref{sec:sims-setup}). In Table~\ref{table:TVCM-YouTube-sims-gains} we present the relative gain in the soft-cache hit ratio, i.e., $g_{Base}(\mathbf{N^{*}})$ vs. hit ratio under no-soft caches scenario. The improvement in performance by using soft-caches can be up to $20\%$, and -on average- is higher in \textit{Instance 1} where the content catalogue is larger (cf. Table~\ref{table:youtube-instances}). The gains are similar in other simulated scenarios we tested; with parameters $Q = \{5, 10, 20, 50\}$ and $TTL=\{0.5, 1, 5, 20\}min$.

Placing contents with the U-aware optimal policy $\mathbf{N_{SCH1}^{*}}$ gives similar gains as in the $\mathbf{N^{*}}$ case in the simulated scenarios. In light of the synthetic results, this perhaps suggests that the content relation graph for these YouTube instances more closely resemble the popularity proportional case, rather than the random.

\begin{table}[!h]
\centering
\caption{Gains in cache hit ratio in the YouTube scenarios.}\label{table:TVCM-YouTube-sims-gains}
\begin{tabular}{c|cc|cc}
{}&\multicolumn{2}{c|}{Instance 1}&\multicolumn{2}{c}{Instance 2}\\
{}					& $Q=5$ 			& $Q=50$			& $Q=5$ 			& $Q=50$	\\
\hline
$TTL=1min$ 			&{$11\%$}		&{$17\%$}		&{$12\%$}		&{$13\%$}\\
$TTL=20min$			&{$20\%$}		&{$19\%$}		&{11$\%$}		&{$7\%$}
\end{tabular}
\end{table}

\section{Discussion}\label{sec:discussion}
Our initial results suggest that soft cache hits could be a promising way to make edge caching scale, opening up new interesting operator-user performance tradeoffs. Some limitations and potential extensions of the proposed model are discussed here.

\emph{User-dependent recommendations:} Throughout this work, we have been assuming that the related contents for a requested content item $i$, and their related utilities depend only on item $i$, and not on the user that requested it. In a sense, this relates to \emph{item-item} collaborative filtering, where a new/alternative item is recommended based on its similarity with the requested one. Item-item recommendations have been claimed to offer some advantages compared to \emph{user-user} collaborative filtering~\cite{Collab-Filter-Sigmetrics2016}. Nevertheless, one user might be less happy than another, with the same alternative content. On the modeling side, one could take this into account by making $u_{ij}$ a random variable and using its expected value $E[u_{ij}]$ in the objective functions of Section~\ref{sec:theory}. 
Finally, on the recommendation side, a recommendation system could actually combine both types of collaborative filtering to make better recommendation. This would lead to different $\mathbf{U}$ graphs per user (or user clusters), whose integration and impact on our framework is part of future work.

\emph{Generalization of $\mathbf{U}$ graph:} For simplicity, in our analysis we assumed that related contents bring the same amount of utility ($1$ in case 1, and $c < 1$ in case 2). In general, different related contents might bring different amounts of utility. We could generalize our model by assuming a \emph{Case 3} where $u_{ii}=1, u_{ij}\in [0,1)~i\neq j$.
As in Case 2, if a user requesting content $i$, accesses (before $i$) any content $j\in \mathcal{R}_{i}$, she will be satisfied $u_{ij}\in(0,1)$ (less than $1$). She will keep on requesting $i$ till time $T$, but will \textit{not} accept any other related content\footnote{An alternative approach would be to keep requesting every cache encountered for potentially better related content. However, we believe this might put a high burden on the battery of the UE and the UE-SC traffic.}. Contrary to Case 2, however, the value of the utility $u_{ij}$ (to be contributed at the objective function) is not known a priori, since we cannot know a priori which content $j\in R_{i}$ will be accessed. 
One can still derive a closed form objective function with appropriate conditioning on all possible $j$, but we defer elaborating on this scenario for future work.


\emph{Generic mobility:} Although it would be quite hard to relax the independent mobility assumption (using traces in simulations, where most such assumptions break, tends to be the de facto way of testing this) the identical contact rate assumption could be relaxed. E.g., in the context of exponential meetings, it has been shown that heterogeneous rates could be approximated with their mean, either asymptotically or as a bound~\cite{TechRep-heterogeneous-spreading}. 

\emph{Soft Cache Hits for Femto-caching:} The proposed approach of soft cache hits and alternative content recommendations could apply equally well to more traditional caching frameworks that do not allow any delay, as is the popular femto-caching framework~\cite{femto}. The relation between users and small cells that each user can access is captured by a bipartite graph, and the control variables $x_{kj}$ define whether a content $k$ is stored in a cache $j$. In the case of $\mathbf{U}$ as in Case 1, if some user requests content $i$, and the small cells in her range are $\mathcal{G} \subseteq \mathcal{M} $, the hit probability is given by
\begin{equation}
 1- \Pi_{j \in \mathcal{G}} \cdot \Pi_{k = 1}^{K} \left(1-x_{kj}\right)^{u_{ik}},
\end{equation}
instead of $ 1- \Pi_{j \in \mathcal{G}} \left(1-x_{ij}\right)$,
in the original femtocaching case (see~\cite{femto} for more details).

\section{Conclusions}\label{sec:conclusion}
In this paper, we have proposed the idea of \emph{soft cache hits} for mobile edge caching systems with delay tolerance, where a user request can sometimes be (partially) satisfied, even if the original content is not available locally, by recommending some related contents. We have formulated and analyzed the performance of such a joint system, and derived the optimal related content aware cache placement. Our theoretical analysis and initial evaluation suggest that significant performance gains can be achieved, even with simple modifications to the baseline system. Furthermore, our results suggest that the structure of the content relation graph plays an important role on the actual achievable performance.

\bibliographystyle{ieeetr}

\end{document}